\documentclass[11pt]{article}
\usepackage{amsmath,amsfonts,epsf}
\usepackage{amssymb}
\usepackage{graphicx}
\usepackage{grffile}
\input epsf
\usepackage[nosort]{cite}
\usepackage{amssymb,amsmath,amsthm}

\usepackage{epstopdf}
\usepackage{enumerate}
\usepackage{tensor}
\usepackage{mathrsfs}
\usepackage{amsmath}
\usepackage{comment}

\usepackage[english]{babel}
\usepackage{braket}
\usepackage{dcolumn}
\usepackage{bm}
\usepackage[hidelinks]{hyperref}
\usepackage{slashed}

\usepackage{comment}

\makeatletter\@addtoreset{equation}{section}\makeatother

\def\bra#1{\mathinner{\langle{#1}|}}
\def\ket#1{\mathinner{|{#1}\rangle}}

\def\bZ {\mathbb{Z}}

\renewcommand{\title}[1]{\vbox{\center\LARGE{#1}}\vspace{5mm}}
\renewcommand{\author}[1]{\vbox{\center#1}\vspace{5mm}}
\newcommand{\address}[1]{\vbox{\center\em#1}}

\begin{document}

\def\aj{\ref@jnl{AJ}}                   
\def\actaa{\ref@jnl{Acta Astron.}}      
\def\araa{\ref@jnl{ARA\&A}}             
\def\apj{\ref@jnl{ApJ}}                 
\def\apjl{\ref@jnl{ApJ}}                
\def\apjs{\ref@jnl{ApJS}}               
\def\ao{\ref@jnl{Appl.~Opt.}}           
\def\apss{\ref@jnl{Ap\&SS}}             
\def\aap{\ref@jnl{A\&A}}                
\def\aapr{\ref@jnl{A\&A~Rev.}}          
\def\aaps{\ref@jnl{A\&AS}}              
\def\azh{\ref@jnl{AZh}}                 
\def\baas{\ref@jnl{BAAS}}               
\def\bac{\ref@jnl{Bull. astr. Inst. Czechosl.}}
\def\caa{\ref@jnl{Chinese Astron. Astrophys.}}
\def\cjaa{\ref@jnl{Chinese J. Astron. Astrophys.}}
\def\icarus{\ref@jnl{Icarus}}           
\def\jcap{\ref@jnl{J. Cosmology Astropart. Phys.}}
\def\jrasc{\ref@jnl{JRASC}}             
\def\memras{\ref@jnl{MmRAS}}            
\def\mnras{\ref@jnl{MNRAS}}             
\def\na{\ref@jnl{New A}}                
\def\nar{\ref@jnl{New A Rev.}}          
\def\pra{\ref@jnl{Phys.~Rev.~A}}        
\def\prb{\ref@jnl{Phys.~Rev.~B}}        
\def\prc{\ref@jnl{Phys.~Rev.~C}}        
\def\prd{\ref@jnl{Phys.~Rev.~D}}        
\def\pre{\ref@jnl{Phys.~Rev.~E}}        
\def\prl{\ref@jnl{Phys.~Rev.~Lett.}}    
\def\pasa{\ref@jnl{PASA}}               
\def\pasp{\ref@jnl{PASP}}               
\def\pasj{\ref@jnl{PASJ}}               
\def\rmxaa{\ref@jnl{Rev. Mexicana Astron. Astrofis.}}%
\def\qjras{\ref@jnl{QJRAS}}             
\def\skytel{\ref@jnl{S\&T}}             
\def\solphys{\ref@jnl{Sol.~Phys.}}      
\def\sovast{\ref@jnl{Soviet~Ast.}}      
\def\ssr{\ref@jnl{Space~Sci.~Rev.}}     
\def\zap{\ref@jnl{ZAp}}                 
\def\nat{\ref@jnl{Nature}}              
\def\iaucirc{\ref@jnl{IAU~Circ.}}       
\def\aplett{\ref@jnl{Astrophys.~Lett.}} 
\def\apspr{\ref@jnl{Astrophys.~Space~Phys.~Res.}}
\def\bain{\ref@jnl{Bull.~Astron.~Inst.~Netherlands}} 
\def\fcp{\ref@jnl{Fund.~Cosmic~Phys.}}  
\def\gca{\ref@jnl{Geochim.~Cosmochim.~Acta}}   
\def\grl{\ref@jnl{Geophys.~Res.~Lett.}} 
\def\jcp{\ref@jnl{J.~Chem.~Phys.}}      
\def\jgr{\ref@jnl{J.~Geophys.~Res.}}    
\def\jqsrt{\ref@jnl{J.~Quant.~Spec.~Radiat.~Transf.}}
\def\memsai{\ref@jnl{Mem.~Soc.~Astron.~Italiana}}
\def\nphysa{\ref@jnl{Nucl.~Phys.~A}}   
\def\physrep{\ref@jnl{Phys.~Rep.}}   
\def\physscr{\ref@jnl{Phys.~Scr}}   
\def\planss{\ref@jnl{Planet.~Space~Sci.}}   
\def\procspie{\ref@jnl{Proc.~SPIE}}   

\let\astap=\aap
\let\apjlett=\apjl
\let\apjsupp=\apjs
\let\applopt=\ao

\newtheorem{theorem}{Theorem}

\newtheorem{conjecture}{Conjecture}
\begin{titlepage}
\begin{center}
\vskip 1cm

\title{A Normal Form for Single-Qudit Clifford+$T$ Operators}

\author{Akalank Jain$^{1}$, Amolak Ratan Kalra$^{2}$, Shiroman Prakash$^{1}$}

\address{${}^1$Department of Physics and Computer Science, Dayalbagh Educational Institute, Dayalbagh, Agra, India 282005}
\address{${}^2$Department of Computer Science, University of Calgary, Alberta, Canada}

\end{center}

\begin{abstract}
 We propose a normal form for single-qudit gates composed of Clifford and $T$-gates for qudits of odd prime dimension $p\geq 5$. We prove that any single-qudit Clifford+$T$ operator can be re-expressed in this normal form in polynomial time. We also provide strong numerical evidence that this normal form is unique. Assuming uniqueness, we are able to use this normal form to provide an algorithm for exact synthesis of any single-qudit Clifford+$T$ operator with minimal $T$-count. 
\end{abstract}

\vfill

\end{titlepage}

\eject \tableofcontents

\clearpage

\section{Introduction}

A problem of central importance in quantum computation is expressing (or approximating) an arbitrary unitary operator as a finite sequence of quantum gates taken from a given universal set. Within the context of fault-tolerant quantum computing, this universal set is dictated by requirements of fault-tolerance.
When an exact decomposition is possible the process of deriving the sequence is called \textit{exact synthesis}. When no exact sequence of gates exists, the unitary operator can be approximated to some arbitrary accuracy $\epsilon$ -- which we refer to as \textit{approximate synthesis}. Efficient and approximate synthesis of arbitrary single and multi-qubit unitary gates of has seen tremendous progress using ideas from number theory in the past several years\cite{doi:10.1063/1.4927100, kliuchnikov2013fast, bocharov2013efficient, PhysRevLett.110.190502, PhysRevLett.112.140504, RossSelinger2014}. 

Here, we address the problem of exact and approximate synthesis for single-qudit gates for qudits of generic odd prime dimension $p$. Our motivation for focusing our attention on qudits of odd prime dimension \cite{PhysRevA.57.127} includes interesting results in qudit magic state distillation \cite{PhysRevA.92.022312, campbell2014enhanced, ACB, CampbellAnwarBrowne, Howard, doi:10.1080/03081079.2017.1308361,doi:10.1080/03081079.2015.1076405,2020golay, jain2020qutrit, PhysRevA.101.010303}, and the fact that magic, as a resource, may be easier to quantify when one restricts to qudits of odd dimension. \cite{Veitch_2012, Veitch_2014, nature, wang2018efficiently, wang_2019}. Qudits also naturally arise in topological field theories -- the Hilbert space of an $SU(2)_k$ Chern-Simons theory on a torus is $k+1$-dimensional \cite{witten1989} -- see, e.g., \cite{fliss2020knots, schnitzer2020n} for discussion related to quantum computing.

A ``canonical example'' of a set of gates that can be implemented in a fault-tolerant manner is the Clifford group, which can be defined for qubits as well as qudits of arbitrary odd-prime dimension $p$, this can be generated by three gates: the Hadamard gate $H$, the phase gate $S$, and the controlled-sum gate $cSUM$. While there are several schemes for implementing the Clifford group in a fault-tolerant manner; the generators of the Clifford group do not form a universal gate set. To be able to approximate an arbitrary unitary operator, one must supplement the Clifford group with another generator. For qubits, a very natural choice for this generator is the single-qubit gate known in the literature as the $T$ gate, or sometimes, as the $\pi/8$-gate. 

A natural generalization for qubit $T$ gates to qudits of arbitrary odd-prime dimension was presented in \cite{HowardVala}. These qudit $T$-gates can be implemented via state-injection via a certain class of qudit magic states; and, in \cite{ CampbellAnwarBrowne, campbell2014enhanced}, very attractive magic state distillation routines for these qudit magic states were constructed. Given these results, we feel that it is important to address the problem of gate synthesis using the qudit version of the Clifford+$T$ group. 
$H$ and $S$ are single-qudit gates, while $CSUM$ is a multi-qudit gate. We refer to the group generated by $H$ and $S$ as the single-qudit Clifford group, and the group generated by $H$, $S$ and $T$ as the single qudit Clifford+$T$ group. It is well known that, if one has access to a multi-qudit gate (such as CSUM) capable of generating entanglement without ancillas \cite{brylinski2001universal, PhysRevLett.89.247902, PhysRevLett.94.230502}, the problem of synthesizing an arbitrary mult-qudit gate can be reduced to the problem of synthesizing an arbitrary single-qudit gate. Here, we therefore address the problem of exact and approximate synthesis of single-qudit gates for qudits of generic odd prime dimension $p$, using the single-qudit Clifford+$T$ group. 

\section{Normal forms}
As reviewed in  \cite{  GilesSelinger2013}, \textit{normal forms} play a central role in the efficient approximation problem for qubits. 

Let us first recall the definition of a normal form. Given a subset $S$ of a group $G$, a \textbf{word} over $S$ is defined to be a finite string of the form:
\begin{equation}
    a_1^{\pm 1}a_2^{\pm 1}\ldots a_N^{\pm 1}
\end{equation}
where each $a_i \in S$. When the $a_i$ are multiplied in the order specified by the word, one obtains a group element $g \in G$; and we say that the word represents $g$. A word is said to be  \textbf{reduced}, if it contains no substring of the form $a a^{-1}$ for some $a \in S$. 

A subset $S$ of $G$ \textbf{generates}  a group $G$ if the set of all words over $S$ contains all elements of $G$. A \textbf{normal form} for a group $G$, is a subset $\mathcal N$ of  words over $S$, such that each element of $G$ is represented by exactly one word in $\mathcal N$.

The Matsumoto-Amano normal form \cite{matsumoto2008representation, GilesSelinger2013} for qubits is essentially the statement that every single qubit Cliiford+$T$ operator can be written in the following normal form:
\begin{equation}
(T|\epsilon)(HT|SHT)^{*}\mathcal{C}.
\end{equation}
The notation used here is that of regular expressions; $\epsilon$ denotes the empty string and $\mathcal{C}$ denotes an arbitrary single-qubit Clifford operator.  The Matsumoto Amano normal form for qubits has the following three properties:
\begin{enumerate}
\item \textit{Existence} -- Any Clifford+$T$ operator can be represented by a word in Matsumoto-Amano normal form. 
\item \textit{Uniqueness} -- There is exactly one word in Matsumoto-Amano normal form representing any Clifford+$T$ operator. 
\item \textit{T-optimality} -- Of all words representing a single-qubit Clifford+$T$ operator, the Matsumoto Amano normal form uses the minimum number of $T$ gates.
\end{enumerate}
By virtue of these properties, the Matsumoto-Amano normal form directly translates into an efficient algorithm for exact synthesis of single-qubit Clifford+$T$ gates, as described in \cite{GilesSelinger2013}. Approximate synthesis can be carried out by the Solvay-Kitaev algorithm; but a more efficient approach based on number theory has also been developed, (which makes also crucial use of the Matsumoto-Amano normal form) in \cite{RossSelinger2014}.
 
A normal form for  single-qutrit Clifford+$T$ operators with above properties, generalizing the Matsumoto-Amano normal form for qubits, was developed in \cite{PhysRevA.98.032304, ross}.  Here, we show that the normal form for single-qudit gates composed of Clifford and $T$ gates can be extended to qudits of arbitrary odd prime dimension $p\geq 5$.

We are able to provide a polynomial time algorithm that transforms any string of $H$, $S$ and $T$ operators into a word this normal form, without increasing the $T$-count. We also provide strong numerical evidence that this form is unique for arbitrary primes $p$. Assuming uniqueness, we are able show that our normal form is $T$-optimal, and present an algorithm for exact synthesis of any single-qudit Clifford+$T$ operator.

A weakness of the present work is our inability to provide a proof of uniqueness. We believe that it should, in principle, be straightforward to provide a proof of uniqueness for any particular prime $p$, along the lines of the proof given in \cite{PhysRevA.98.032304}; but that this would require computations that are too tedious to perform by hand. It may be possible to automate the strategy of \cite{PhysRevA.98.032304} for proving uniqueness to produce a computer-aided proof, but we do not attempt that here. 

\section{\label{sec:level1}Clifford+$T$ group in $p$ dimensions}
In this section, we review the definition of the Clifford+$T$ group in $p$ dimensions.

\subsection{Clifford group in $p$ dimensions}
Recall that the natural generalization of the qubit Pauli operators $X$ and $Z$ to $p$ dimensional qudits \cite{Gottesman1999} is:
\begin{eqnarray}
X \ket{k} & = & \ket{k+1} \\
Z \ket{k} & = & \omega^k \ket{k}.
\end{eqnarray}
Here, and in what follows, letters $j$ and $k$ denote elements of the finite field $\mathbb Z_p$, with multiplication and addition defined modulo $p$ and $$\omega=e^{2\pi i /p}$$ is the $p$th root of unity. The single qudit Heisenberg-Weyl displacement group $\mathcal D$ is the group generated by $X$ and $Z$.

The single-qudit Clifford group $\mathcal C$ is defined to be the set of operators which map the single-qudit Heisenberg-Weyl displacement group to itself under conjugation \cite{Gottesman1999}:
\begin{equation}
    \mathcal C = \{ C \in SU(p) | C \mathcal D C^\dagger = \mathcal D \}.
\end{equation}

Since we are primarily concerned with the action of these gates on density matrices, overall phases are unphysical. Hence all gates should be thought of as elements of $PSU(p)$ rather than $U(p)$ or $SU(p)$.  However, for our purposes, it is more convenient to consider the single-qudit Clifford group as a subgroup of $SU(p)$, so we will choose to define the generators such that they have determinant $1$. It can be shown that the single-qudit Clifford group is generated by the single-qudit unitaries $S$ and $H$, which, for $p>2$, we define as follows:
\begin{eqnarray}
S & = & \sum_{j=0}^{p-1} \omega^{j(j+1) 2^{-1}} \ket{j} \bra{j} \\
H & = & \delta_p\frac{1}{\sqrt{p}} \sum_{j=0}^{p-1} \sum_{k=0}^{p-1} \omega^{jk} \ket{j}\bra{k}.
\end{eqnarray}
In the above expression, $2^{-1}$ denotes the inverse of $2$ in the finite field $\mathbb Z_p$. (For example, if $p=5$, $2^{-1}=3$.)  We included an overall phase $\delta_p$ in the definition of the Hadamard gate, given by \begin{equation}
    \delta_p = \begin{cases} 1 & p \equiv 1 \mod 8 \\
    -i & p \equiv 3 \mod 8 \\
    -1 & p \equiv 5 \mod 8 \\
    i & p \equiv 7 \mod 8.
    \end{cases}
\end{equation}
With this choice of overall phase, using the results of, e.g., \cite{eigenvectorsH}, one can check that $\det H=1$. One can also check that, for $p\geq 5$, $\det S=1$. 

Let us define $\epsilon_p$ as follows:
\begin{equation}
    \epsilon_p = \begin{cases} 1 & p \mod 4 = 1 \\ i & p \mod 4 =- 1.
    \end{cases}
\end{equation}
Then the overall phase $\delta_p$ can be written as,
\begin{equation}
    \delta_p=\left(\frac{2}{p}\right)_L \epsilon_p,
\end{equation}
where $\left( \frac{2}{p} \right)_L$ is the Legendre symbol (see Appendix \ref{vf-appendix}) from number theory. It is a well-known result in number theory  (see, e.g., \cite{apostol1976introduction}) that 
\begin{equation}
    \sum_{k=0}^{p-1} \omega^{2k^2} = \delta_p \sqrt{p},
\end{equation}
so we can rewrite $H$ as
\begin{equation}
    H = \frac{1}{p} \left( \sum_{m=0}^{p-1} \omega^{2m^2} \right) \sum_{j=0}^{p-1} \sum_{k=0}^{p-1} \omega^{jk} \ket{j}\bra{k}. \label{H-tilde}
\end{equation}

Observe that if $\det C=1$ then $\det \omega^n C=1$ for any $n$. Therefore, we should also add the generator $\omega$ to the list of generators of the Clifford group if we are considering the Clifford group to be a subgroup of $SU(p)$.

\subsection{Clifford group and $SL(2,\bZ_p)$}

The Clifford group is closely related to the semidirect product of $SL(2,\mathbb Z_p)$ and $\mathbb Z_p \otimes \mathbb Z_p$, as shown explicitly in \cite{Appleby}. (See also \cite{Wootters1987, Gross}.) 

The group $SL(2,\bZ_p)$ is the set of matrices 
\begin{equation} 
\begin{pmatrix} 
a & b \\ c & d 
\end{pmatrix}
\end{equation}
whose entries,  $a,~b,~c,~d \in \bZ_p$, with determinant 1.  We will denote elements of $SL(2,\bZ_p)$ by overhats.

$SL(2,\bZ_p)$ can be generated from the two matrices $\hat{S}$ and $\hat{H}$ given by:
\begin{equation}
\hat{S}=\begin{pmatrix} 1 & 0 \\ 1 & 1 \end{pmatrix}, ~\hat{H} = \begin{pmatrix} 0 & -1 \\ 1 & 0 \end{pmatrix}.
\end{equation}
In total, $SL(2,\bZ_p)$ contains $p(p^2-1)$ elements.

Let us now present the relation between the Clifford group and the semi-direct product of $SL(2,\mathbb Z_p)$ and $\mathbb Z_p \otimes \mathbb Z_p$, derived in \cite{Appleby}. We first define the map $V(\hat{F})$ as follows. Let $$\hat{F}=\begin{pmatrix} a & b \\ c & d \end{pmatrix},$$ then
\begin{equation}
V({\hat{F}})=
\begin{cases}
\left( \frac{-2b}{p} \right)_L \epsilon_p \frac{1}{\sqrt{p}} \displaystyle \sum_{j,k=0}^{p-1} \omega^{2^{-1} b^{-1}(ak^2-2jk+dj^2)}\ket{j}\bra{k} & b\neq 0 \\
\displaystyle \left( \frac{a}{p} \right)_L \sum_{k=0}^{p-1}\omega^{2^{-1} ack^2} \ket{ak}\bra{ k} & b= 0. \\
\end{cases} \label{vf}
\end{equation}
Here we have modified the definition of $V(\hat{F})$ given in \cite{Appleby} to include an overall phase factor. This phase factor guarantees that $\det V(\hat{F})=1$.
We show that $V(\hat{F})$ defines a homeomorphism $V:SL(2,\mathbb Z_p) \rightarrow SU(p)$ in Appendix \ref{vf-appendix}.

We then define the map $D:\mathbb Z_p \otimes \mathbb Z_p \rightarrow PSU(p)$, as follows:
\begin{equation}
D_{(x,z)} = \omega^{2^{-1} x z} X^xZ^z.
\end{equation} 
This is a homeomorphism to $PSU(p)$, but not $SU(p)$, as discussed in the Appendix.

In particular, up to an overall phase, any element of the Clifford group $C$ can be written as
\begin{equation}
C \sim D_{\vec{\chi}}V(\hat{F}),
\end{equation} 
where $\vec{\chi} = \begin{pmatrix} x \\ z \end{pmatrix}$ is an element of $\mathbb Z_p^2$, and  is an element of $SL(2,\mathbb Z_p).$ (The overall phase is guaranteed to be a power of $\omega$ with our definition of $V(\hat{F})$.) These functions obey,
\begin{equation}
D_{\vec{\chi_1}}V({\hat{F}_1}) D_{\vec{\chi_2}}V({\hat{F}_2}) \sim D_{\vec{\chi}_1+\hat{F}_1\vec{\chi}_2} V({\hat{F}_1 \hat{F}_2}).\end{equation}
where $\sim$ denotes equality up to an overall phase (which is, again some power of $\omega$).

The generators $S$ and $H$ can be written as $D_{(0,2^{-1})}V_{\hat{S}}$. and $V_{\hat{H}}$ respectively. The Clifford group, defined as a subgroup of $PSU(p)$, therefore has $p^3 (p^2-1)$ elements. These correspond to $p^4 (p^2-1)$ elements of $SU(p)$.

Note that the set of elements of $SL(2,\mathbb Z_p)$ with $b=0$ form a subgroup of $SL(2,\mathbb Z_p)$.

\subsection{$T$ gates in $p$ dimensions}
In various models of fault-tolerant quantum computation, particularly the magic state model (see \cite{msd} and, e.g., \cite{2020golay} and references therein), Clifford unitaries are assumed to be fault-tolerant, and can therefore be implemented with zero cost. However, the single-qudit Clifford group is a finite group, and therefore cannot be used to approximate an arbitrary single-qudit unitary. It is therefore necessary to supplement the Clifford group with another generator to promote it to an infinite group that is dense in $SU(p)$ and can therefore be used to approximate arbitrary single-qudit unitaries. 

It was shown in \cite{CampbellAnwarBrowne} that any non-Clifford gate can be used to promote the Clifford group into a dense subgroup of $SU(p)$, following \cite{Nebe2001, Nebe:2006:SCI:1208720}.  (We checked this explicitly for the $T$-gate defined below, using the  criteria of \cite{sawicki2017universality, sawicki2017criteria} for small primes ($p=3$, $5$, and $7$).) 

For qubits, the most natural choice for such a gate is the $T$ gate (which has a variety of interesting properties, described in, e.g., \cite{Goldengates}). The analogues of $T$ gates in higher odd prime dimensions were defined in \cite{HowardVala} and \cite{CampbellAnwarBrowne}. 
These are defined to be diagonal gates that solve the following equation:
\begin{equation}
    T_{(m,z,\gamma)} XT_{(m,z,\gamma)}^\dagger = \omega^m D_{(1,z)} V\left(\hat{S}^\gamma\right).
\end{equation}
The set of gates above form a group, isomorphic to $\mathbb Z_p^3$, that contains $p^3$ elements, out of which $p^2(p-1)$ are non-Clifford gates. It is easy to see that any gate of the form $T_{(m,z,\gamma)}$ is related to $T_{(0,0,\gamma)}$ by multiplication by Clifford gates. For $p\geq 5$ we choose to define the \textit{canonical} $T$ gate as:
\begin{equation}
T \equiv T_{(0,0,1)}= \sum_k \omega^{k^3 6^{-1}} \ket{k}\bra{k}. 
\end{equation}
One can check that $\det T=1$ and $T^p=1$. Let us define the group $\mathcal T = \{\mathbb I, T, T^2,T^3,T^4, \ldots T^{p-1}\}$, and the set $\mathcal T'=\mathcal T/\mathbb I = \{T,T^2, \ldots T^{p-1}\}$.

Explicitly, for $p=5$, $T$ is given by:
\begin{eqnarray}
T_{(5)}& = & \left(
\begin{array}{ccccc}
 1 & 0 & 0 & 0 & 0 \\
 0 & \omega & 0 & 0 & 0 \\
 0 & 0 & \omega^3 & 0 & 0 \\
 0 & 0 & 0 & \omega^2 & 0 \\
 0 & 0 & 0 & 0 & \omega^4 \\
\end{array}
\right)
\end{eqnarray}

Implementing a $T$ gate would typically require state injection by distilled magic states, and is therefore expensive. To quantify the cost of synthesizing an operator, we therefore need to define the $T$-count of a word.  

Before we define the $T$-count, we should ask whether $T^n$ is related to $T$ by a Clifford transformation. Define 
\begin{equation}
\hat{F}_a=\begin{pmatrix} a & 0 \\ 0 & a^{-1}\end{pmatrix}. \label{Fa}
\end{equation} 
Then we can check that 
\begin{equation}
V_{\hat{F}_a}^\dagger T V_{\hat{F}_a} = T^{a^{3}}.
\end{equation}
For primes such that, $p \mod 3 =2$, the equation $a^3 \equiv n \mod p$ always has a solution for any nonzero $n$. So for these primes, $T^n = V_{\hat{F}_a}^\dagger T V_{\hat{F}_a}$ for some $a$. For $p \mod 3=1$, this is not true, and $T^n$ may not be related to $T$ by a Clifford transformation. However, we expect that, implementing the gate $T^n$ via state injection, following, e.g., \cite{campbell2014enhanced, CampbellAnwarBrowne}, should be no more costly than implementing $T$. 

We therefore define the \textbf{$T$-count} of a word as follows. Any element of the Clifford+$T$ group can be represented as a word over the generating sets $\mathcal C$ and $\mathcal T'$. We define the $T$-count of the word to be the total number of elements of the set $\mathcal T'$ it contains, i.e., the total number of disjoint $T^{n}$ operators it contains (where $n\neq 0$). We emphasize that in our definition, any power of $T$ can be implemented with $T$-count equal to 1. For example, the string $HT^2HT^3H$ has $T$-count $2$.

\section{A normal form for single-qudit Clifford+$T$ gates}

Here we define a normal form, closely following \cite{GilesSelinger2013, PhysRevA.98.032304}. As in \cite{GilesSelinger2013, PhysRevA.98.032304}, we must first define the following subsets of the Clifford group: $\mathcal H$, $\mathcal H'$ and $\mathcal P$. 

Let $\mathcal{\hat{L}}$ denote the subset of $SL(2,\mathbb Z_p)$ consisting of all lower-triangular matrices, i.e., matrices of the form,  $\hat{F}=\begin{pmatrix} a & 0 \\ c & d \end{pmatrix}$. $\mathcal{\hat{L}}$ is generated by $\hat{S}$ and any single matrix of the form $\hat{F}_a$ defined in equation \eqref{Fa}, where $a$ is any primitive root modulo $p$. $\mathcal{\hat{L}}$ contains $p(p-1)$ elements. 

We define $\mathcal P$ to be the subgroup of $\mathcal C$ consisting of all elements of the form $D_{\chi}V(\hat{F})$ where $\hat{F}\in \mathcal L$. It is easy to see that $\mathcal P$ can be generated by $S$, $X$, $V({\hat{F}_a})$, where $a$ is any primitive root modulo $p$. $\mathcal P$ contains $p^3(p-1)$ elements not including overall phases. (Including overall phases, which can only be powers of $\omega$,  $\mathcal P$ contains $p^4(p-1)$ elements.)

It is easy to see from direct computation, that there are $p+1$ left cosets of  $\mathcal L$ in $SL(2,\mathbb Z_p)$, which can each be represented by matrices of the form:
\begin{equation}
\begin{pmatrix}
0 & -1 \\
1 & d 
\end{pmatrix}~\text{, and }
\begin{pmatrix}
1 & 0 \\
0 & 1 
\end{pmatrix},
\end{equation}
for $d=0, \ldots, p-1$. These can be written as $\hat{H}_d=\hat{S}^d \hat{H}$ and $\mathbb I$.
Each of these left cosets correspond to a left coset of $\mathcal P$, each of which can be represented simply by: $H_d \equiv S^d H$ and $\mathbb I$. 

Let us define 
\begin{equation}\mathcal H = \left\{\mathbb I,~ H, ~SH,~ \ldots S^{p-1}H \right\}.\end{equation}
Let us also define
\begin{equation}
\mathcal H ' = \mathcal H \backslash \mathbf{1}.
\end{equation}

Using these components, we propose the following normal form for elements of the single-qudit Clifford+$T$ group:
\begin{equation}
(\mathcal T) (\mathcal H' \mathcal T')^* \mathcal H \mathcal P. \label{normalform}
\end{equation}
The notation used here is a combination of hybrid notation of regular expressions and sets used in \cite{GilesSelinger2013, PhysRevA.98.032304}. For example, $(\mathcal H' \mathcal T')$ represents a string consisting of any element from the set $\mathcal H'$ followed by any element from the set $\mathcal T$.

Explicitly, any word written in the above normal form, can be written as follows:
\begin{equation}
T^{m_0} (H_{d_1}T^{m_1})(H_{d_2}T^{m_2})\ldots((H_{d_n}T^{m_n}))C \label{explicit-normal-form}
\end{equation}
where $C$ is a Clifford operator, $0\leq m_0\leq p-1$, and for $i>0$, $0 \leq d_i\leq p-1$ and $1 \leq m_i \leq p-1$. If $m_0=0$ the $T$-count of the above string is $n$; if $m_0>0$, the $T$-count of the above string is $n+1$. 

It is convenient to consider our normal-form as a string of ``elementary syllables''. We define the \textit{left-most syllable} of a word in the above form as follows.
\begin{enumerate}
  \item If $m_0\neq 0$, and $n=0$, the left-most syllable of $M$ is $T^{m_0}C$.
  \item If $m_0\neq 0$ and $n\neq 0$ the left-most syllable of $M$ is $T^{m_0}H_{d_1}T^{m_1}$
  \item If $m_0=0$ and $n=0$, the left-most syllable is $C$.
  \item If $m_0=0$ and $n \neq 0$ the left-most syllable of $M$ is $H_{d_1}T^{m_1}$. 
\end{enumerate}
Examples: the left-most syllable of $T^2SHTS^2HT^2S$ is $T^2$ and the left-most syllable of $SHTS^HT^2S$ is $SHT=H_1T$. Suppose the left-most syllable of $M$ is $L$. We define the second-left-most syllable of $M$ as the left-most syllable of $L^{-1}M$, and so-on. 


\section{Existence and T-optimality}
Following \cite{GilesSelinger2013} and \cite{PhysRevA.98.032304}, we provide a proof of existence and $T$-optimality of the qudit Matsumoto-Amano normal form. 

The definitions and discussion in the previous section imply the following relations between $\mathcal C$, $\mathcal H$, $\mathcal P$, and $\mathcal H'$:
\begin{eqnarray}
\mathcal{C} & = & \mathcal{H}\mathcal{P} \\
\mathcal{C}\backslash  \mathcal{P} & = & \mathcal H' \mathcal{P} \label{two} \\
\mathcal{P}\mathcal{H'} & \subseteq &  \mathcal H' \mathcal P  \label{three}
\end{eqnarray}

We also observe that
\begin{theorem} \label{PT}
\begin{equation}
\mathcal{P}\mathcal{T} =  \mathcal{T}\mathcal{P}. \label{four}
\end{equation} 
\end{theorem}
\begin{proof}
The generators of $\mathcal P$ are $S$, $V(\hat{F}_a)$ and $X$. 
First observe that $S$ clearly commutes with $T$. We also have $XT = \omega^{-6^{-1}} TXS^{p-1}$. Finally observe $V_{F_a}T=T^nV_{F_{a}}$, where $n=a^{-3}$, from which the result follows. \end{proof}

Using these facts, we now prove that any Clifford+$T$ operator can be written in the normal form \eqref{normalform}. 

Any word representing an element of the Clifford+$T$ can be written as 
\begin{equation}
M=C_n T^{a_n} C_{n-1} T^{a_{n-1}} \ldots C_1 T^{a_1}C_0,
\end{equation} 
where $C_i$ are Clifford operators. We now apply the following two step process to $M$:
\begin{itemize}
\item For $0<i<n$: if $C_i \in \mathcal P$, then we can use Theorem \ref{PT} to replace $T^{a_{i+1}}C_iT^{a_i}C_{i-1}$ by $T^{a_{i}'} C'$ to obtain an equivalent word, without increasing the number of $\mathcal T$ gates required. 

\item Use $T^p=1$ to eliminate all instances of $T^p$.
\end{itemize}

We then simplify the resulting Clifford operators occurring in the string sandwiched between $T$ operators, and repeat this process until we obtain an expression of the form $$M=C_{n'} \mathcal T C_{n'-1} \ldots C_1 \mathcal T C_0$$ where each $C_i \in \mathcal C/\mathcal{P}$ for $n'>i>0$, $C_0 \in \mathcal C$ and $C_{n'} \in \mathcal C$.

Because $\mathcal{C}/\mathcal{P} =  \mathcal H' \mathcal P$, (and $\mathcal{C}=\mathcal H \mathcal P$ for the leftmost operator), using relations \eqref{two}, \eqref{three} and \eqref{four}, we can push any elements of $\mathcal P$ occurring in any  $C_i$ to the far right of the expression as follows:
\begin{eqnarray}
\mathcal C' \mathcal T \mathcal C' \ldots & = & \mathcal H' \mathcal P \mathcal T \mathcal H' \mathcal P \ldots\\
                                          & = & \mathcal H' \mathcal T \mathcal P  \mathcal H' \mathcal P\ldots \\
                                          & \subseteq & \mathcal H' \mathcal T  \mathcal H' \mathcal P \ldots
\end{eqnarray}
so that each $C_i \in \mathcal H'$, for $i \neq 0$, $i \neq n$. The leftmost operator $C_n$ satisfies $C_n \in \mathcal H$ and the rightmost operator $C_0$  satisfies $C_0 \in \mathcal H \mathcal P$.

The final result is an expression of the form
\begin{equation}
M = \mathcal H \mathcal T \mathcal H' \mathcal T \mathcal H' \ldots \mathcal T \mathcal H \mathcal P .
\end{equation}
which is equivalent to the normal form given in equation \eqref{normalform}.

The above procedure proof directly translates into an efficient (polynomial time) algorithm for rewriting any word representing a Clifford+$T$ operators into the proposed normal form. During the process of conversion of any word representing a  Clifford+$T$ operator into the normal form \eqref{normalform} via this algorithm, the $T$-count may decrease or remain unchanged. Suppose one begins with a word that represents a single-qudit Clifford+$T$ operator using the least possible $T$-count. We could rewrite it in the normal form \eqref{normalform} using the procedure outlined above, without increasing its $T$-count. Hence we have shown \textit{existence} and a weak-form of \textit{T-optimality}.

Below, we will also provide substantial numerical evidence that the representation of any single qudit operator by the normal form \eqref{normalform} is unique. Assuming this is true, the output of the algorithm above will always result in the word representing the given Clifford+$T$ operator using the least possible $T$-count.

\section{Uniqueness and exact synthesis}

In this section we present evidence that the normal form defined in \eqref{normalform} is unique -- i.e., that any element of the single-qudit Cliford+$T$ group can be represented by exactly one word in the normal form \eqref{normalform}. Our strategy is based on the proofs of uniqueness for qubits given in \cite{GilesSelinger2013}, and relies on the concept of a least-denominator exponent, which we define below. In particular, we present a conjecture relating the least-denominator exponent of an operator, written as an unitary matrix, to its $\mathcal H'$- count. This conjecture also translates into an exact synthesis algorithm. 

Note, however, that, unlike the case of qubits we have not provided a complete algebraic characterization of Clifford+$T$ operators. We believe that this is not possible. Indeed, in the mathematics literature, it is known that the existence of a simple normal form such as ours (which essentially gives rise to a tree-like structure) is usually incompatible with a simple algebraic characterization of the group (although there are a handful of exceptions to this rule).\footnote{We thank Peter Sarnak for discussions on this point.} 

\subsection{Algebraic preliminaries}
Let us define the following rings:
\begin{itemize}
  \item $\mathbb Z[1/p] = \{\frac{a}{p^n} | ~a \in \mathbb Z,~ n \in \mathbb N  \}$.
  \item $\mathbb Z[\omega] = \{ a_0 +a_1 \omega + a_2 \omega^2 + \ldots a_{p-2}\omega^{p-2}|~a_i \in \mathbb Z \}$. 
  \item $\mathbb Z[\omega,1/p] = \{ a_0 +a_1 \omega + a_2 \omega^2 + \ldots a_{p-2}\omega^{p-2}|~a_i \in \mathbb Z[1/p] \}$. 
\end{itemize}

If $a \in \mathbb Z$, we define $\bar{a} \equiv a \mod p$ to denote the corresponding element in $\mathbb Z_p$. We define the parity map $P: \mathbb Z[\omega] \rightarrow \mathbb Z_p$ as follows:
\begin{equation}
P( a_0 +a_1 \omega + a_2 \omega^2 + \ldots a_{p-2}\omega^{p-2}) \equiv (\bar{a}_1+\bar{a}_2+\bar{a}_3 +\ldots).
\end{equation}
Let us define 
\begin{equation}
\chi=1-\omega.
\end{equation} 
The parity map is a ring homomorphism. It can be thought of as a map onto the equivalence classes induced by the equivalence relation $a \sim b$ if $(a-b) = c \chi$, where $a,~b,~c \in \mathbb Z[\omega]$. 

We define the following \textit{denominator exponent} functions of $y \in \mathbb Z[\omega, 1/p]$: Let $\eta \in \mathbb Z[\omega]$ satisfy $p \Big| \eta^n$ for some $n$. Then a denominator exponent of $y$ relative to $\eta$ is any non-negative integer $k$ such that $\eta^k y \in \mathbb Z[\omega]$. The \textit{least denominator exponent} of $y$ with respect to $\eta$ is the smallest such $k$, and is denoted as $d_\eta(y)$. For example, $d_{\chi}(\frac{1}{p})=p-1$. We will only use denominator exponents relative to $\chi$. 

We define the \textit{denominator exponent} of a matrix $M$ with entries in $\mathbb Z[\omega, 1/p]$ to be any non-negative value of $k$ such that all the entries of $\chi^k M$ are in $\mathbb Z[\omega]$. We define \textit{the least denominator exponent} of a matrix $M$ to be the smallest such value, which we denote as $d(M)$. 

Let $k$ be a denominator exponent of $M$; then $\chi^k M$ has entries in $\mathbb Z[\omega]$. We define 
\begin{equation}
P_k(M) \equiv P( \chi^k M),
\end{equation} 
where $P(\chi^k M)$ acts on each of the entries of $\chi^k M$.

\subsection{Evidence for uniqueness}
Notice that, the entries of $H$ (as presented in equation \eqref{H-tilde}), $S$ and $T$ are elements of $\mathbb Z[\omega,1/p]$. Therefore any element of the Clifford+$T$ group, (thought of as a subgroup of $SU(p)$), is a $p\times p$ matrix with entries in $Z[\omega,1/p]$. 
 
Let us define the $H'$-count of an Clifford+$T$ operator in normal form as the number of $\mathcal H'$ operators that it contains. Explicitly, for the word given in equation \eqref{explicit-normal-form}, the $H'$-count can be computed as follows. If its right-most Clifford operator $C \in \mathcal P$, the $H'$-count of the word in equation \eqref{explicit-normal-form} is $n$. If $C \not\in P$, the $H'$-count of a word the word in equation \eqref{explicit-normal-form} is $n+1$. The $\mathcal H'$-count is clearly related to the $T$-count in a straightforward way.

Through numerical experimentation we obtain the following conjecture:
\begin{conjecture} Let $M$ be a $p\times p$ special unitary matrix $M$ representing a Clifford+$T$ operator. The least denominator exponent $k$ of $M$ is related to the $H'$-count $h(M)$  via \label{uniqueness-conjecture}
\begin{equation}
k= \begin{cases} 
\lceil nh(M)/2 \rceil+n & h(M) \geq 1 \\
0 & h(M)=0.
\end{cases} 
\end{equation}
where $n=\lceil p/3 \rceil-1$.
\end{conjecture}
We tested this conjecture using a large number of random strings generated in our normal form with $H'$-count $\leq 20$, for all primes $p\leq 29$. We believe it would not be hard, in principle, to prove this conjecture for any given $p$, by explicit calculation, but it may be tedious.

The above conjecture implies that two words expressed in our normal form with different $H$'-counts represent different elements of the Clifford+$T$ group. 

We also checked, numerically, that two distinct words in our normal form with the same $H'$-count represent distinct Clifford+$T$ operators. In principle, we also believe that, one can determine the left-most syllable of the word representing a Clifford+$T$ operator $M$ from the parity map $P_k(M)$, via rules analogous to those given in Appendix A of \cite{PhysRevA.98.032304} for qutrits and Figure 2 of \cite{GilesSelinger2013} for qubits. However, these rules appear to be rather complicated, so we do not present them here. 

\subsection{Exact synthesis}
The above conjecture, if true, translates into a simple exact-synthesis algorithm. Suppose a matrix $M \in SU(p)$ is a candidate Clifford+$T$ operator. The following algorithm can be used to determine its first $n=\lceil p/3 \rceil-1$ left-most syllables
\begin{enumerate}
\item Check if $M$ is a Clifford operator $C$ using a finite look-up table. If so, the left-most syllable of $M$ is $C$ and exact synthesis is complete. Otherwise carry out the following steps. 

\item Check if all the entries of $M$ are in $\mathbb Z[\omega,1/p]$.  If so, calculate the least-denominator exponent $k$ of $M$, and move on to the next step. (If not, then $M$ is not a Clifford+$T$ operator and exact synthesis fails.)

\item Conjecture \ref{uniqueness-conjecture} determines the $H'$-count of $M$ from $k$. If $M$ can be expressed in our normal form, there are only a finite number of possibilities for its $n=\lceil p/3 \rceil-1$ left-most syllables: call each of these $L_i$. Calculate the least-denominator exponent $d_\chi\left((L_i)^{-1}M\right)$ for each $L_i$.  

\begin{enumerate}
\item If $d_\chi\left((L_i)^{-1}M\right) \geq k$ for all $L_i$ then $M$ is not a Clifford+$T$ operator and exact synthesis fails. 
\item If $d_\chi\left((L_i)^{-1}M)\right<k$ for any $L_i$ then the first $n=\lceil p/3 \rceil-1$ left-most syllables of $M$ are given by $L_i$. (It is a consequence of Conjecture \ref{uniqueness-conjecture} that not more than one $L_i$ can satisfy this condition.) 
\end{enumerate}
\end{enumerate}
If $M$ is not a Clifford operator, we repeat this algorithm on $(L_i^{-1}) M$ to obtain the next $n=\lceil p/3 \rceil-1$ left-most syllables. If repeat this process, eventually we will obtain a Clifford operator, or the algorithm will fail, indicating that $M$ is not an element of the Clifford+$T$ group.

\section{Approximating elements of $SU(p)$}
To approximate an element of $SU(p)$ that is not in the Clifford+$T$ group, using Clifford+$T$ operators, the Solvay-Kitaev algorithm \cite{solvaykitaev} can be used. A gate $U \in SU(p)$ is said to $\epsilon$-approximated by a gate $G$ if $d(G,U)\equiv ||G-U||\equiv \sup_{|\psi|=1}||G-U||\psi \leq \epsilon$.

\begin{theorem}
Let $\epsilon > 0$. There exists $U \in SU(p)$ whose $\epsilon$-approximation by a Clifford+$T$ circuit requires an operator of $T$ count at least:
\begin{equation}
    n\gtrsim \frac{\ln\left(\frac{B}{A}\right)+(p^{2}-1)\ln\left(\frac{1}{\epsilon}\right)}{\ln(p(p-1))}, \label{t-bound}
\end{equation} where 
\begin{equation}A=\frac{p^4\pi^{(p^2-1)/2}}{(p(p-1)-1)\Gamma\left(\frac{p^2-1}{2}\right)},~\text{and } B=\pi^{\frac{(p-1)(p-2)}{2}}\sqrt{2^{p-1}p}\prod_{k=1}^{p-1}\frac{1}{k!}.
\end{equation}
\end{theorem}
\begin{proof}
To prove this theorem, we use a volume-counting argument similar to \cite{ross}. The total number of operators in canonical form with $T$-count $t\geq 1$ is given by:
\begin{equation}
m(t)=p^{2}((p-1)p)^{t+1}(1+p)^{2}.
\end{equation}
Note that this count includes overall phases (which must be powers of $\omega$). From this, we see that the total number of different operators in normal form with total $T$-count at most $n$ is \begin{equation}
N(n) =\sum_{t=0}^n m(t) =\frac{p^4 \left(\left(p^2-1\right)^2 ((p-1) p)^n-p^3+p\right)}{p(p-1)-1}.
\end{equation}
The total volume of $SU(p)$ is given in \cite{sphere} to be:
\begin{equation}
\text{vol}(SU(p))=\sqrt{2^{p-1}p}\pi^{\frac{(p-1)(p-2)}{2}}\prod_{k=1}^{p-1}\frac{1}{k!}.
\end{equation}
If $\epsilon$ is sufficiently small, then each Clifford+$T$ operators allows one to $\epsilon$-approximate operators in $SU(p)$ contained with the volume of a $d=(p^2-1)$-dimensional ball of radius $\epsilon$, which is:
\begin{equation}
    \text{vol}(B_{d}(\epsilon)) =\frac{\pi^{d/2}}{\Gamma(d/2+1)}{\epsilon}^{d}.
\end{equation}
To approximate any element to $\epsilon$ accuracy the complete volume of $SU(p)$ needs to be covered. The would require at least:
$$N \gtrsim \frac{\text{vol}\left(SU(p)\right)}{\text{vol}\left(B_{p^2-1}(\epsilon)\right)}$$ distinct Clifford+$T$ operators; from which we conclude that $n$ must satisfy: 
\begin{equation}
\frac{p^4 \left(\left(p^2-1\right)^2 ((p-1) p)^n-p^3+p\right)}{p(p-1) -1}\frac{\pi^{d/2}}{\Gamma(p^{2}/2+1)}{\epsilon}^{p^{2}-1} \geq \sqrt{p2^{p-1}}\pi^{\frac{(p-1)(p-2)}{2}}\prod_{k=1}^{p-1}\frac{1}{k!}.
\end{equation}
For large $n$ the  expression simplifies to:
\[
\frac{p^{4}((p(p-1))^{n})}{-1+p(p-1)}\frac{\pi^{(p^2-1)/2}}{\Gamma((p^2-1)/2)}{\epsilon}^{p^{2}-1} \gtrsim \sqrt{p2^{p-1}}\pi^{\frac{(p-1)(p-2)}{2}}\prod_{k=1}^{p-1}\frac{1}{k!}
\]
We use constants $A$ and $B$  as defined above  to further simplify the expression to obtain:
\begin{equation}
n\gtrsim \frac{\ln\left(\frac{B}{A}\right)+(p^{2}-1)\ln\left(1/{\epsilon}\right)}{\ln(p(p-1))}
\end{equation}
\end{proof}
Let us emphasize that equation \ref{t-bound} is an asymptotic  \textit{lower-bound} on the minimal $T$-count required to be able to approximate any unitary operator to a given accuracy $\epsilon$.

\section{Acknowledgements}
SP would like to thank Prof. P. S. Satsangi for guidance. The authors would like to thank N. J. Ross for his comments on an earlier version of the manuscript. This research is supported in part by a DST INSPIRE Faculty award, DST-SERB Early Career Research Award (ECR/2017/001023) and MATRICS grant (MTR/2018/001077). 
\appendix

\section{Relation between $SL(2,\mathbb Z_p)$ and the Clifford group}
\label{vf-appendix}
Here we show that $V:SL(2,\mathbb Z_p) \rightarrow SU(p)$ defined in \eqref{vf} is a homomorphism. To do this we must show that $V(\hat{F}) V(\hat{F}')=V(\hat{F}\hat{F}')$, where $F= \begin{pmatrix} a & b \\ c & d\end{pmatrix}$ and $F'=\begin{pmatrix} a' & b' \\ c' & d' \end{pmatrix}$. 

In order to do this, we will need to use the following result from number theory
\begin{equation}
    \sum_{j=0}^{p-1}\omega^{aj^2} = \left( \frac{a}{p}\right)_L \epsilon_p \sqrt{p},
\end{equation}
known as a quadratic Gauss sum. $\left( \frac{a}{p}\right)_L$ denotes what is known as the Legendre symbol in number theory. $\left( \frac{a}{p}\right)_L=+1$ if $a$ is a quadratic residue modulo $p$, $\left( \frac{a}{p}\right)_L-=-1$ if $a$ is not a quadratic residue modulo $p$, and $\left( \frac{a}{p}\right)_L=0$ if $p$ divides $a$. 

There are three cases:
\begin{enumerate}
    \item Case 1: $b=b'=0$. The result \begin{equation}
        V(F)V(F') = \left( \frac{aa'}{p} \right)_L \sum_{k=0}^{p-1} \omega^{2^{-1} a a' (a'c+a^{-1}c')} \ket{a a'k}\bra{k} = V(FF'),
    \end{equation} 
    follows easily from equation \eqref{vf} and the fact that the Legendre symbol is multiplicative: $$\left(\frac{a}{p}\right)_L\left(\frac{a'}{p}\right)_L= \left(\frac{aa'}{p}\right).$$
    \item Case 2: $b \neq 0$, $b'= 0$. In this case
    \begin{eqnarray*}
    V(F)V(F') & = & \left( \frac{-2ba'}{p} \right)_L \epsilon_p \frac{1}{\sqrt{p}} \displaystyle \sum_{j,k=0}^{p-1} \omega^{2^{-1} b^{-1}(ak^2-2jk+dj^2)}\ket{j}\bra{k}  \sum_{k'=0}^{p-1}\omega^{2^{-1} a'c'k'^2} \ket{a'k'}\bra{ k'} \\ 
     & = &
    \left( \frac{-2ba'}{p} \right)_L \epsilon_p \frac{1}{\sqrt{p}} \displaystyle \sum_{j,k=0}^{p-1} \omega^{2^{-1} b^{-1}(aa'^2k^2-2ja'k+dj^2)} \omega^{2^{-1} a'c'k^2} \ket{j}\bra{k} \\
      & = &
    \left( \frac{-2ba'}{p} \right)_L \epsilon_p \frac{1}{\sqrt{p}} \displaystyle \sum_{j,k=0}^{p-1} \omega^{2^{-1} b^{-1}a'(aa'k^2-2jk+a'^{-1}dj^2)} \omega^{2^{-1}b^{-1} a' bc'k^2} \ket{j}\bra{k} \\
          & = &
    \left( \frac{-2b{a'}^{-1}}{p} \right)_L \epsilon_p \frac{1}{\sqrt{p}} \displaystyle \sum_{j,k=0}^{p-1} \omega^{2^{-1} b^{-1}a'\left((aa'+bc')k^2-2jk+a'^{-1}dj^2\right)}  \ket{j}\bra{k} \\ & = & V(FF').
    \end{eqnarray*}
   We used the fact that $\left(\frac{a'}{p} \right)_L = \left(\frac{a'^{-1}}{p} \right)_L$ in the last line.
    
    \item Case 3: $b\neq 0$, $b' = 0$.
    This case follows from calculations analogous to the previous case.
    \item Case 4: $b\neq 0$ and $b'\neq 0$
    \begin{eqnarray*}
    V(F)V(F') & = & \left( \frac{4bb'}{p} \right)_L \epsilon_p^2 \frac{1}{p} \displaystyle \sum_{j,k,j',k'=0}^{p-1} \omega^{2^{-1} b^{-1}(ak^2-2jk+dj^2)}\ket{j}\bra{k}   \omega^{2^{-1} b'^{-1}(a'k'^2-2j'k'+d'j'^2)}\ket{j'}\bra{k'} \\
     & = & \left( \frac{bb'}{p} \right)_L \epsilon_p^2 \frac{1}{p} \displaystyle \sum_{j,k,k'=0}^{p-1} \omega^{2^{-1} b^{-1}(ak^2-2jk+dj^2)}   \omega^{2^{-1} b'^{-1}(a'k'^2-2kk'+d'k^2)}\ket{j}\bra{k'} \\
       & = & \left( \frac{bb'}{p} \right)_L \epsilon_p^2 \frac{1}{p} \displaystyle \sum_{j,k,k'=0}^{p-1} \omega^{2^{-1} b^{-1}(dj^2) + b'^{-1}(a'k'^2)}  \omega^{2^{-1} b'^{-1}(-2kk'+d'k^2)+2^{-1} b^{-1}(ak^2-2jk)}\ket{j}\bra{k'}
       \\
       & = & \left( \frac{bb'}{p} \right)_L \epsilon_p^2 \frac{1}{p} \displaystyle \sum_{j,k,k'=0}^{p-1} \omega^{2^{-1} b^{-1}b'^{-1}(b'dj^2+ ba'k'^2)}  \omega^{2^{-1} b'^{-1}b^{-1}\left((ab'+bd')k^2-2(bk'+b'j)k\right)}\ket{j}\bra{k'} \\
       & = & \left( \frac{bb'}{p} \right)_L \epsilon_p^2 \frac{1}{p} \displaystyle \sum_{j,k,k'=0}^{p-1} \omega^{2^{-1} b^{-1}b'^{-1}(b'dj^2+ ba'k'^2)}  \omega^{2^{-1} b'^{-1}b^{-1}\left(b''k^2-2(bk'+b'j)k\right)}\ket{j}\bra{k'}
    \end{eqnarray*}
    We now consider two sub-cases:
    \begin{enumerate}
        \item  If $b''=ab'+bd'=0$, then we can write \begin{equation}
            F F' = \begin{pmatrix}
            -bb'^{-1} & 0 \\
           -a'b^{-1}-db'^{-1} & -b'b^{-1}
            \end{pmatrix}
        \end{equation}
        and $V(F)V(F')$ reduces to:
    \begin{eqnarray}
    & = & \left( \frac{bb'}{p} \right)_L \epsilon_p^2 \frac{1}{p} \displaystyle \sum_{j,k,k'=0}^{p-1} \omega^{2^{-1} b^{-1}b'^{-1}(b'dj^2+ ba'k'^2)}  \omega^{2^{-1} b'^{-1}b^{-1}\left(-2(bk'+b'j)k\right)}\ket{j}\bra{k'} \\ & = &
    \left( \frac{-bb'^{-1}}{p} \right)_L  \displaystyle \sum_{j}^{p-1} \omega^{2^{-1} b^{-1}b'^{-1}(b^2db'^{-1}+ ba')j^2}  \ket{-bb'^{-1}j}\bra{j} 
    \\ & = &
    \left( \frac{-bb'^{-1}}{p} \right)_L  \displaystyle \sum_{j}^{p-1} \omega^{-2^{-1} b b'^{-1}(-db'^{-1}+ a'b^{-1})j^2}  \ket{-bb'^{-1}j}\bra{j} \\
    & = &  V(FF')
    \end{eqnarray}
    
    \item If $b''\neq 0$, we can evaluate the quadratic Gauss sum over $k$ by completing the square to obtain,
    \begin{eqnarray}
      & = & \left( -2^{-1}\frac{b''}{p} \right)_L \epsilon_p \frac{1}{\sqrt{p}} \displaystyle \sum_{j,k'=0}^{p-1} \omega^{2^{-1} (bb')^{-1}(b'dj^2+ ba'k'^2)}  \omega^{-2^{-1} (bb'b'')^{-1}\left(bk'+b' j \right)^2}\ket{j}\bra{k'} \\
      & = & \left( -2^{-1}\frac{b''}{p} \right)_L \epsilon_p \frac{1}{\sqrt{p}} \displaystyle \sum_{j,k=0}^{p-1} \omega^{2^{-1} (bb')^{-1}\left(b'dj^2+ ba'k^2 -(b'')^{-1}\left(bk+b' j \right)^2\right)}\ket{j}\bra{k}
      \\
      & = & \left( -2^{-1}\frac{b''}{p} \right)_L \epsilon_p \frac{1}{\sqrt{p}} \displaystyle \sum_{j,k=0}^{p-1} \omega^{2^{-1} (b'')^{-1}\left(a a' k^2+\frac{a b' d j^2}{b}+\frac{a'b d' k^2}{b'}-\frac{b' j^2}{b}-\frac{b k^2}{b'}+dd' j^2-2 j k \right)}\ket{j}\bra{k} \\
      & = & V(FF')
    \end{eqnarray}
    
    \end{enumerate}
\end{enumerate}

Note that, while $V$ defines a homomorphism to $SU(p)$, the map $D_{\vec{\chi}}$ does not preserve phases, and is map to $PSU(p)$: $D:\mathbb Z_p \otimes \mathbb Z_p \rightarrow PSU(p)$. To see this note that
\begin{eqnarray}
D_{\chi_1} D_{\chi_2}& = & \omega^{2^{-1}(x_1 z_1+x_2z_2)}X^{x_1}Z^{z_1}X^{x_2}Z^{z_2} \\ & = & 
\omega^{2^{-1}(x_1 z_1+x_2z_2)+z_1x_2}X^{x_1+x_2}Z^{z_1+z_2} \\
& = & 
\omega^{2^{-1}\left(-x_1z_2+z_1x_2\right)}D_{\chi_1+\chi_2}
\end{eqnarray}

\bibliographystyle{ssg}
\bibliography{qudit}

\end{document}